\documentclass[12pt]{iopjournal}

\usepackage[T1]{fontenc}
\usepackage[utf8]{inputenc}
\usepackage{microtype}
\usepackage{amsmath,amssymb,amsthm,mathtools,bm}
\usepackage{booktabs}
\usepackage{enumitem}
\usepackage[numbers,sort&compress]{natbib}
\hypersetup{
  pdftitle={Exact Finite-Channel Schur Moments of the Wigner--Smith Matrix with a Tunnel Barrier},
  pdfauthor={Israel F. Araujo}
}

\newtheorem{theorem}{Theorem}
\newtheorem{proposition}[theorem]{Proposition}
\newtheorem{corollary}[theorem]{Corollary}

\theoremstyle{remark}

\newcommand{\Tr}{\operatorname{Tr}}
\newcommand{\E}{\mathbb{E}}
\newcommand{\U}{\mathrm{U}}
\newcommand{\CUE}{\mathrm{CUE}}
\newcommand{\dd}{\mathrm{d}}
\newcommand{\id}{\mathbb{I}}
\newcommand{\PR}{P_R}
\newcommand{\QR}{Q_R}
\newcommand{\Qzero}{Q_0}
\newcommand{\tauD}{\tau_{\mathrm D}}
\newcommand{\tauW}{\tau_{\mathrm W}}

\makeatletter
\renewcommand{\articletype}[1]{{\vspace*{-8mm}\noindent \large\sffamily Journal of Physics A: Mathematical and Theoretical}
\vspace*{7mm}\noindent\reversemarginpar
\marginpar{\vspace{-3mm}{\color{gray}\hrule}\ \\
{\scriptsize\sffamily\bfseries\MakeUppercase{#1}}}}
\makeatother
\begin{document}

\articletype{Paper}

\title{Exact Finite-Channel Schur Moments of the Wigner--Smith Matrix with a Tunnel Barrier}

\author{Israel F. Araujo}
\affil{data cybernetics ssc GmbH, Landsberg am Lech, 86899, Germany\\}
\affil{Author to whom any correspondence should be addressed.\\}
\email{araujo@data-cybernetics.com}

\keywords{Wigner--Smith time-delay matrix, chaotic scattering, random matrix theory, Schur polynomials, circular unitary ensemble, Toeplitz determinants, tunnel barriers, finite-channel corrections}

\begin{abstract}
We consider the Wigner--Smith time-delay matrix of a chaotic cavity with broken time-reversal symmetry, a finite number $M$ of equivalent open channels, and a uniform tunnel barrier of reflection probability $R$.  Exact random-matrix results for nonideal coupling are known at the level of the joint matrix distribution, whereas explicit finite-channel formulas for Schur moments are much less developed.  We show that every Schur moment of the nonideal Wigner--Smith matrix factorizes into its ideal-coupling inverse-Laguerre value and a barrier factor given by a circular unitary ensemble (CUE) average of a Schur polynomial.  A dual-Cauchy generating function reduces the latter to a Toeplitz determinant with a rational symbol, placing the barrier problem within the classical theory of exact CUE characteristic-polynomial ratios.

For one-row partitions, specializing the exact CUE ratio formula gives a closed generating function for all barrier coefficients and therefore all finite delay-time moments in the allowed range $n\le M$.  The exact ratio formula decomposes into a sector visible to the semiclassical expansion and an additional contribution proportional to $R^{M+1}$, which is invisible to every algebraic order in $1/M$.  For the second moment, this contribution exactly reproduces the exponentially small term in the known random-matrix expression for $\langle\tauW^2\rangle$ that is absent from the semiclassical expansion.  We also derive an exact finite sum for all one-column partitions and an exact complement symmetry for Schur averages inside the $M\times M$ square.  The new contribution therefore lies not in the underlying CUE ratio identities, which are classical, but in combining them with the exact nonideal Wigner--Smith mapping and extracting the resulting finite-channel time-delay formulas and corrections to semiclassical Schur-moment conjectures.
\end{abstract}

\section{Introduction}

The Wigner--Smith matrix
\begin{equation}
  Q=-\mathrm{i}\hbar S^{\dagger}\frac{\partial S}{\partial E}
  \label{eq:wigner-smith}
\end{equation}
encodes temporal information in quantum scattering.  Its eigenvalues are the proper delay times, while the normalized trace
\begin{equation}
  \tauW=\frac{1}{M}\Tr Q
\end{equation}
is the Wigner time delay.  For chaotic cavities with ideal contacts, random-matrix theory (RMT) maps the inverse proper delay times to a Laguerre ensemble \cite{brouwer1997,brouwer1999}.  Thus proper-delay statistics are negative-moment statistics of Laguerre eigenvalues \cite{mezzadrisimm2011,martinez2014}, while exact partition-indexed ideal-coupling averages are known \cite{novaes2015,novaes2022}.  The author's master's dissertation \cite{araujo2017thesis} combined this inverse-Laguerre viewpoint with symmetric-function and determinant-generating-function methods for ideal-coupling delay-time observables.

The nonideal-coupling problem is structurally richer.  Tunnel barriers correlate the scattering matrix and the Wigner--Smith matrix and introduce a new reflection parameter.  Exact RMT formulations of the joint matrix distribution are available \cite{grabsch2018,grabsch2019corrigendum}, but explicit finite-$M$ formulas for general symmetric moments are much harder to obtain.  A semiclassical treatment by Novaes and Kuipers \cite{novaeskuipers2023} produced systematic asymptotic series and motivated several conjectures for Schur moments at arbitrary reflection probability.  A recent review emphasizes two unresolved features \cite{novaes2026}: first, the semiclassical expansion of $\langle\tauW^2\rangle$ misses an exact term proportional to $R^{M+1}$, which is exponentially small for large $M$; second, simple conjectured relations among Schur moments, including barrier independence for self-conjugate partitions, remain without an RMT proof.

We address this problem by combining the exact nonideal-to-ideal RMT mapping with symmetric-function methods and established exact formulas for CUE characteristic-polynomial ratios.  We retain from the earlier ideal-coupling viewpoint \cite{araujo2017thesis} the inverse-Laguerre organization and the use of determinant generating functions, but no technical step below requires access to the dissertation: the ideal input is stated explicitly and supported by the journal literature.  The exact mapping factorizes each Schur moment into the known ideal inverse-Laguerre value and a CUE barrier average; a dual-Cauchy generating function then turns the latter into a rational Toeplitz determinant.  The required Toeplitz and CUE product/ratio formulas belong to a well-developed mathematical theory \cite{day1975,basorforrester1994,conrey2005,bumpgamburd2006}.  Our contribution is to apply this structure to the tunnel-barrier Wigner--Smith problem and extract its finite-$M$ consequences.

Our main results are as follows.
\begin{enumerate}[leftmargin=*,itemsep=2pt]
\item Every finite Schur moment at nonideal coupling factorizes into the ideal inverse-Laguerre moment and a universal CUE barrier factor; see Theorem~\ref{thm:factorization}.
\item A multivariate determinant generating function for the barrier factors is exactly a Toeplitz determinant with rational symbol; see Proposition~\ref{prop:toeplitz}.  This identifies the problem with a classical class of exact CUE product/ratio averages \cite{day1975,basorforrester1994,conrey2005,bumpgamburd2006}.
\item All one-row moments are generated by the closed expression in Theorem~\ref{thm:rowGF}, obtained by specializing the exact two-over-two CUE ratio formula.  Their exact coefficients contain an additional sector proportional to $R^{M+1}$, which lies beyond all algebraic orders in $1/M$ and is the finite-rank ``swap'' sector of that ratio average.
\item All one-column moments admit the exact finite sum in Theorem~\ref{thm:column}.  The row and column formulas reproduce the known exact second Wigner-delay moment, including its $R^{M+1}$ term.
\item The CUE barrier averages obey an exact complement symmetry inside the $M\times M$ square; see Theorem~\ref{thm:complement}.
\end{enumerate}

We restrict the paper to the unitary symmetry class ($\beta=2$), equivalent channels, and a uniform tunnel barrier.  This setting is sufficient to expose the finite-channel sector that is invisible to the semiclassical $1/M$ expansion.  General-$\beta$ extensions, naturally involving Jack polynomials \cite{novaes2022}, are left for future work.

\section{Model and notation}
\label{sec:model}

We consider $M$ equivalent channels and a tunnel barrier with reflection probability
\begin{equation}
  0\le R<1,
  \qquad r=\sqrt{R}.
\end{equation}
The ideal-coupling dwell time is denoted by $\tauD$.  Throughout the paper $\lambda=(\lambda_1,\lambda_2,\ldots)$ denotes an integer partition, $|\lambda|=\sum_i\lambda_i$ its weight, $\ell(\lambda)$ its length, and $\lambda'$ its conjugate.  Schur polynomials are denoted by $s_\lambda$; we use the standard symmetric-function conventions of Ref.~\cite{macdonald1995}.

For later use we define the shifted products
\begin{align}
 [M]^\lambda &= \prod_{(i,j)\in\lambda}(M+j-i),
 \\
 [M]_\lambda &= \prod_{(i,j)\in\lambda}(M-j+i).
 \label{eq:shiftedproducts}
\end{align}
Here $1^M$ denotes $M$ variables all equal to one.  The hook-content formula gives
\begin{equation}
  s_\lambda(1^M)=\frac{d_\lambda}{|\lambda|!}[M]^\lambda,
  \label{eq:hookcontent}
\end{equation}
where $d_\lambda$ is the dimension of the irreducible representation of the symmetric group $S_{|\lambda|}$ indexed by $\lambda$.

For ideal contacts, the Schur moments of the Wigner--Smith matrix are \cite{novaes2015,novaes2022}
\begin{equation}
 \E\,s_\lambda(\Qzero)
 =
 (M\tauD)^{|\lambda|}
 \frac{d_\lambda}{|\lambda|!}
 \frac{[M]^\lambda}{[M]_\lambda}.
 \label{eq:ideal-schur}
\end{equation}
The moment is finite when the Young diagram fits within an $M\times M$ square, i.e., when $\lambda_1\le M$ and $\ell(\lambda)\le M$.

The exact nonideal-contact mapping is most naturally stated for the symmetrized Wigner--Smith matrix \cite{grabsch2018,grabsch2019corrigendum},
\begin{equation}
 Q_s=S^{1/2}QS^{-1/2}
 =-\mathrm{i}\hbar\,S^{-1/2}\frac{\partial S}{\partial E}S^{-1/2}.
 \label{eq:Qs}
\end{equation}
Because $Q_s$ is similar to $Q$, the two matrices have the same proper delay times.  All Schur-polynomial observables considered below are therefore unchanged when $Q$ is replaced by $Q_s$.  From this point onward, $\QR$ and $\Qzero$ denote the symmetrized representatives at nonideal and ideal coupling, respectively.

For equivalent channels in the unitary class, the exact mapping of Ref.~\cite{grabsch2018} can be written, after a Cayley parametrization of the ideal scattering matrix and an immaterial sign change $U\mapsto-U$, as
\begin{equation}
  \QR=A\Qzero A,
  \label{eq:Qmap}
\end{equation}
with
\begin{equation}
  A^2=\PR(U)
  =(1-R)(\id+rU)^{-1}(\id+rU^\dagger)^{-1},
  \qquad U\sim\CUE(M).
  \label{eq:PR}
\end{equation}
Here $R=|\overline S|^2=r^2$ is the channel reflection probability.  To connect Eq.~\eqref{eq:PR} explicitly with the coupling parametrization of Ref.~\cite{grabsch2018}, let $\kappa$ denote its channel-coupling constant and introduce the signed reflection amplitude
\begin{equation}
 U=(\id-\mathrm{i}K)(\id+\mathrm{i}K)^{-1},
 \qquad
 r_\kappa=\frac{1-\kappa}{1+\kappa},
 \qquad R=r_\kappa^2.
 \label{eq:cayley-coupling}
\end{equation}
Because $K$ is Hermitian, functional calculus gives
\begin{equation}
 (1-R)(\id+r_\kappa U)^{-1}(\id+r_\kappa U^\dagger)^{-1}
 =\kappa(\id+K^2)(\id+\kappa^2K^2)^{-1}=A^2.
 \label{eq:A2-check}
\end{equation}
Since the Haar measure is invariant under $U\mapsto-U$, all CUE averages are unchanged under $r_\kappa\mapsto-r_\kappa$; hence the remainder of the paper may use the nonnegative convention $r=\sqrt R$ in Eq.~\eqref{eq:PR}.  At ideal coupling, $\Qzero$ is unitarily invariant and independent of $U$ \cite{brouwer1999,grabsch2018}.  These two facts are the starting point of the exact calculation.

\section{Exact factorization of nonideal Schur moments}
\label{sec:factorization}

Define the barrier Schur average
\begin{equation}
 I_\lambda(R,M)
 :=\E_{\CUE(M)}\,s_\lambda\!\left(\PR(U)\right).
 \label{eq:I-def}
\end{equation}

\begin{theorem}[Schur-moment factorization]
\label{thm:factorization}
For every partition $\lambda\subseteq M^M$ (equivalently, for every Schur moment finite at ideal coupling),
\begin{equation}
 \boxed{
 \E\,s_\lambda(\QR)
 =
 \E\,s_\lambda(\Qzero)
 \frac{I_\lambda(R,M)}{s_\lambda(1^M)}
 }
 \label{eq:factorization}
\end{equation}
and therefore
\begin{equation}
 \boxed{
 \E\,s_\lambda(\QR)
 =
 (M\tauD)^{|\lambda|}
 \frac{I_\lambda(R,M)}{[M]_\lambda}.
 }
 \label{eq:factorization-compact}
\end{equation}
\end{theorem}

\begin{proof}
Because $A$ is invertible for $0\le R<1$, the matrices $A\Qzero A$ and $\Qzero A^2$ are similar, and therefore
$s_\lambda(A\Qzero A)=s_\lambda(\Qzero\PR)$.  Conditional on the eigenvalues of $\Qzero$ and $\PR$, unitary invariance allows the relative eigenbasis to be represented by a Haar matrix $V\in\U(M)$.  The standard character-convolution identity, with normalized Haar measure,
\begin{equation}
 \int_{\U(M)}
 s_\lambda(XVYV^\dagger)\,\dd V
 =
 \frac{s_\lambda(X)s_\lambda(Y)}{s_\lambda(1^M)}
 \label{eq:character-convolution}
\end{equation}
then factorizes the conditional average.  Independence of $\Qzero$ and $U$ gives Eq.~\eqref{eq:factorization}; substituting Eqs.~\eqref{eq:ideal-schur} and \eqref{eq:hookcontent} gives Eq.~\eqref{eq:factorization-compact}.
\end{proof}

Equation~\eqref{eq:factorization-compact} isolates all barrier dependence in the compact CUE object $I_\lambda$.  In particular,
\begin{equation}
 I_\lambda(0,M)=s_\lambda(1^M),
\end{equation}
so the ideal-coupling result is recovered immediately.

\section{A rational-Toeplitz generating function}
\label{sec:toeplitz}

Let $L$ be a positive integer and introduce
\begin{equation}
 \Psi_L(\bm t)
 =
 \E_{\CUE(M)}
 \prod_{a=1}^{L}\det(\id+t_a\PR),
 \qquad \bm t=(t_1,\ldots,t_L).
 \label{eq:Psi-def}
\end{equation}
The dual Cauchy identity gives
\begin{equation}
 \prod_{a=1}^{L}\det(\id+t_a X)
 =
 \sum_{\substack{\lambda:\,\lambda_1\le L\\\ell(\lambda)\le M}}
 s_\lambda(X)s_{\lambda'}(t_1,\ldots,t_L),
\end{equation}
and hence
\begin{equation}
 \boxed{
 \Psi_L(\bm t)
 =
 \sum_{\substack{\lambda:\,\lambda_1\le L\\\ell(\lambda)\le M}}
 I_\lambda(R,M)\,
 s_{\lambda'}(\bm t).
 }
 \label{eq:Psi-Schur}
\end{equation}
Thus any target $I_\lambda$ can be recovered from $\Psi_L$ by choosing $L\ge\lambda_1$ and extracting the Schur coefficient of $s_{\lambda'}(\bm t)$.

For a symbol $g$ on the unit circle, write $D_M[g]=\det[g_{j-k}]_{j,k=0}^{M-1}$, where $g_n$ denotes its $n$th Fourier coefficient.  Let $z=e^{\mathrm{i}\theta}$ be an eigenvalue of $U$.  The corresponding eigenvalue of $\PR$ is
\begin{equation}
 p_R(z)
 =
 \frac{1-R}{(1+rz)(1+rz^{-1})}.
 \label{eq:pR}
\end{equation}

\begin{proposition}[Rational Toeplitz representation]
\label{prop:toeplitz}
The generating function \eqref{eq:Psi-def} is the Toeplitz determinant
\begin{equation}
 \boxed{
 \Psi_L(\bm t)=D_M[g_L],
 \qquad
 g_L(z)=\prod_{a=1}^{L}\bigl(1+t_a p_R(z)\bigr).
 }
 \label{eq:toeplitz-main}
\end{equation}
For $0<R<1$ and each $a$, define $\rho_a$ by
\begin{equation}
 \rho_a+\rho_a^{-1}
 =
 \frac{1+R+t_a(1-R)}{\sqrt R},
 \label{eq:rho-def}
\end{equation}
choosing the branch $\rho_a\to r$ as $t_a\to0$.  Then
\begin{equation}
 \boxed{
 g_L(z)=
 \prod_{a=1}^{L}
 \frac{(z+\rho_a)(z+\rho_a^{-1})}
 {(z+r)(z+r^{-1})}.
 }
 \label{eq:rational-symbol}
\end{equation}
\end{proposition}

\begin{proof}
The Heine--Szeg\H{o} identity converts the CUE eigenvalue average of a product of single-eigenvalue factors into the corresponding Toeplitz determinant, giving Eq.~\eqref{eq:toeplitz-main}.  Equation~\eqref{eq:rational-symbol} follows by matching the constant and $z+z^{-1}$ coefficients in $1+t_ap_R(z)$.
\end{proof}

The factorized form \eqref{eq:rational-symbol} extends to $R=0$ by continuity.  For $0<R<1$ the symbol is rational, so the finite-$M$ determinant can be evaluated using the exact formulas of Day \cite{day1975} and determinant reformulations such as Basor--Forrester \cite{basorforrester1994}.  Equivalently, after factorization it is an average of products and ratios of CUE characteristic polynomials, for which exact finite-rank formulas were derived by Conrey--Forrester--Snaith and Bump--Gamburd \cite{conrey2005,bumpgamburd2006}.  Thus Proposition~\ref{prop:toeplitz} is a reduction of the time-delay problem to known CUE machinery, not a new Toeplitz identity.  For a fixed target partition $\lambda$ one may take $L=\lambda_1$, so the number of auxiliary factors is set by the partition width; the channel number $M$ remains in the finite-size parameters.

We do not require the fully confluent general determinant below.  Instead, the two basic partition families---one row and one column---can be evaluated in closed form and already reveal the finite-channel structure that is invisible to the perturbative $1/M$ expansion.

\section{One-row partitions and the finite-channel sector beyond the \texorpdfstring{$1/M$}{1/M} expansion}
\label{sec:rows}

For a one-row partition $\lambda=(n)$, $s_{(n)}=h_n$ is the complete homogeneous symmetric polynomial.  Define
\begin{equation}
 H_M(t)
 :=
 \E_{\CUE(M)}\frac{1}{\det(\id-t\PR)}
 =
 \sum_{n=0}^{\infty}I_{(n)}(R,M)t^n.
 \label{eq:H-def}
\end{equation}
For sufficiently small $t$, let $x=x(t)$ denote the branch with $x(0)=1$ satisfying
\begin{equation}
 \boxed{
 t=\frac{(x-1)(1-Rx)}{x(1-R)}.
 }
 \label{eq:x-def}
\end{equation}

\begin{theorem}[Exact one-row generating function]
\label{thm:rowGF}
For $0\le R<1$,
\begin{equation}
 \boxed{
 H_M(t)
 =
 \frac{x^M\left[(1-Rx)^2-R^{M+1}(x-1)^2\right]}
 {(1-R)(1-Rx^2)},
 }
 \label{eq:H-exact}
\end{equation}
where $x=x(t)$ is defined by Eq.~\eqref{eq:x-def}.
\end{theorem}

The theorem follows directly from the classical exact CUE ratio formula.  Indeed, Eq.~\eqref{eq:x-def} implies
\begin{equation}
  1-tp_R(z)
  =x^{-1}\frac{(1+rxz)(1+rxz^{-1})}
  {(1+rz)(1+rz^{-1})},
  \label{eq:row-ratio-factor}
\end{equation}
so that
\begin{equation}
 H_M(t)=x^M\,\E_{\CUE(M)}
 \frac{\det(\id+rU)\det(\id+rU^\dagger)}
 {\det(\id+rxU)\det(\id+rxU^\dagger)}.
 \label{eq:row-cue-ratio}
\end{equation}
The two-over-two CUE ratio formula \cite{conrey2005,bumpgamburd2006}, with both numerator parameters equal to $-r$ and both denominator parameters equal to $-rx$, gives
\begin{align}
 &\E_{\CUE(M)}
 \frac{\det(\id+rU)\det(\id+rU^\dagger)}
 {\det(\id+rxU)\det(\id+rxU^\dagger)} \notag\\
 &\quad=
 \frac{(1-Rx)^2}{(1-R)(1-Rx^2)}
 -\frac{R^{M+1}(x-1)^2}{(1-R)(1-Rx^2)}.
 \label{eq:row-ratio-specialized}
\end{align}
Multiplication by $x^M$ proves Eq.~\eqref{eq:H-exact}.  Appendix~\ref{app:rowproof} gives an independent Toeplitz/Kac--Murdock--Szeg\H{o} derivation.  The result has a useful exact splitting into a perturbative part and a finite-channel correction:
\begin{equation}
 H_M(t)=H_M^{\mathrm{pert}}(t)+H_M^{\mathrm{fc}}(t),
 \label{eq:H-split}
\end{equation}
with
\begin{align}
 H_M^{\mathrm{pert}}(t)
 &=
 \frac{x^M(1-Rx)^2}{(1-R)(1-Rx^2)},
 \\
 H_M^{\mathrm{fc}}(t)
 &=
 -\frac{R^{M+1}x^M(x-1)^2}{(1-R)(1-Rx^2)}.
 \label{eq:H-fc}
\end{align}
At fixed $0<R<1$, $H_M^{\mathrm{fc}}$ is exponentially small as $M\to\infty$.  It is therefore absent to every algebraic order in an ordinary expansion in $1/M$.  Equation~\eqref{eq:row-ratio-specialized} also identifies its exact origin: it is the finite-rank second, or ``swap'', term of the CUE characteristic-polynomial ratio average.  The $R^{M+1}$ sector is therefore not a new random-matrix ratio phenomenon; what is new here is its appearance in the tunnel-barrier time-delay problem and the extraction of its coefficients at all orders.

\subsection{Exact coefficients}

Set
\begin{equation}
 c=\frac{R}{1-R}.
\end{equation}
Applying Lagrange--B\"urmann inversion to Eq.~\eqref{eq:x-def} gives the following finite sums.

\begin{corollary}[All one-row barrier factors]
\label{cor:rowcoeff}
For $n\ge2$,
\begin{equation}
 I_{(n)}(R,M)
 =I_{(n)}^{\mathrm{pert}}(R,M)+\Delta I_{(n)}(R,M),
 \label{eq:row-split}
\end{equation}
where
\begin{equation}
 \boxed{
 I_{(n)}^{\mathrm{pert}}
 =
 \sum_{k=0}^{n}
 \binom{M+n-1}{n-k}
 \binom{n+k-2}{k}
 c^k,
 }
 \label{eq:Ipert}
\end{equation}
and
\begin{equation}
 \boxed{
 \Delta I_{(n)}
 =
 -\frac{R^{M+1}}{(1-R)^2}
 \sum_{k=0}^{n-2}
 \binom{M+n-1}{n-2-k}
 \binom{n+k}{k}
 c^k.
 }
 \label{eq:Inpcorr}
\end{equation}
For $n=1$, $I_{(1)}=M$ and the correction vanishes.
\end{corollary}

The corresponding delay-time Schur moments follow from Theorem~\ref{thm:factorization}.  They are finite for $1\le n\le M$:
\begin{equation}
 \boxed{
 \E\,s_{(n)}(\QR)
 =
 (M\tauD)^n\frac{I_{(n)}(R,M)}{[M]_{(n)}}.
 }
 \label{eq:row-physical}
\end{equation}
Here $[M]_{(n)}=M(M-1)\cdots(M-n+1)$.

The CUE barrier coefficients $I_{(n)}$ in Corollary~\ref{cor:rowcoeff} are well defined for all $n$, but the corresponding delay-time moment in Eq.~\eqref{eq:row-physical} diverges once $n>M$.  This distinction is important when the generating function is used purely as a barrier identity.

The perturbative part has an instructive relation to the one-row (``singleton'') conjecture of Ref.~\cite{novaeskuipers2023}.  To avoid ambiguity below, for a generic variable $a$ write $a^{(j)}=a(a+1)\cdots(a+j-1)$ and $a_{(j)}=a(a-1)\cdots(a-j+1)$.  Multiplying Eq.~\eqref{eq:Ipert} by $(1-R)^n$ and applying a binomial convolution yields
\begin{align}
 &(1-R)^n\E\,s_{(n)}(\QR)\big|_{\mathrm{pert}}
 \nonumber\\
 &\quad=
 \frac{(M\tauD)^n}{n!(M)_{(n)}}
 \sum_{k=0}^{n}(-R)^k\binom{n}{k}
 (M+k)^{(n-k)}(M-n+k)_{(k)},
 \label{eq:NK-sector}
\end{align}
which is precisely the low-degree polynomial sector conjectured from the semiclassical series, with the factor of $\tauD$ restored.  The exact RMT completion is Eq.~\eqref{eq:Inpcorr}.  For fixed integer $M$ the full result remains a polynomial after multiplication by $(1-R)^n$, but it contains an additional high-degree sector beginning at $R^{M+1}$.

\subsection{The first nontrivial example}

For $n=2$ (hence $M\ge2$), define
\begin{equation}
 \mathcal{A}_M(R)
 =M(M+1)-2(M^2-1)R+M(M-1)R^2.
 \label{eq:calAdef}
\end{equation}
Equation~\eqref{eq:row-physical} becomes
\begin{equation}
 \boxed{
 \E\,s_{(2)}(\QR)
 =
 \frac{M\tauD^2}{2(M-1)(1-R)^2}
 \left[\mathcal{A}_M(R)-2R^{M+1}\right].
 }
 \label{eq:s2}
\end{equation}
The term $-2R^{M+1}$ lies beyond all algebraic orders of the semiclassical $1/M$ expansion.

\section{One-column partitions}
\label{sec:columns}

For $\lambda=(1^k)$, $s_{(1^k)}=e_k$.  The generating function is
\begin{equation}
 G_M(t)
 :=
 \E\det(\id+t\PR)
 =D_M[1+t p_R].
 \label{eq:G-def}
\end{equation}
The Fourier coefficients of $p_R$ are
\begin{equation}
 p_R(e^{\mathrm{i}\theta})
 =\sum_{m\in\mathbb Z}(-\sqrt R)^{|m|}e^{\mathrm{i}m\theta},
 \label{eq:poisson-fourier}
\end{equation}
so
\begin{equation}
 G_M(t)=\det(\id+tK_M),
 \qquad
 (K_M)_{ij}=(-\sqrt R)^{|i-j|}.
 \label{eq:KMS}
\end{equation}
Thus $K_M$ is a Kac--Murdock--Szeg\H{o} matrix.

\begin{theorem}[All one-column barrier factors]
\label{thm:column}
For $2\le k\le M$,
\begin{equation}
 \boxed{
 I_{(1^k)}(R,M)
 =
 (1-R)^{k-1}
 \sum_{j=0}^{M-k}
 \binom{M-j}{k}
 \binom{k+j-2}{j}R^j.
 }
 \label{eq:column-barrier}
\end{equation}
For $k=1$, $I_{(1)}=M$.
\end{theorem}

\begin{proof}
The coefficient of $t^k$ in Eq.~\eqref{eq:KMS} is the sum of the $k\times k$ principal minors of $K_M$.  If $1\le i_1<\cdots<i_k\le M$, the corresponding Kac--Murdock--Szeg\H{o} minor is
\begin{equation}
 \det K[i_1,\ldots,i_k]
 =\prod_{a=1}^{k-1}\left(1-R^{i_{a+1}-i_a}\right).
 \label{eq:KMSminor}
\end{equation}
Summing over the gaps and collecting equal powers of $R$ gives Eq.~\eqref{eq:column-barrier}; Appendix~\ref{app:columnproof} provides the details.
\end{proof}

The corresponding delay-time moment is
\begin{equation}
 \boxed{
 \E\,s_{(1^k)}(\QR)
 =
 (M\tauD)^k
 \frac{I_{(1^k)}(R,M)}{[M]_{(1^k)}}.
 }
 \label{eq:column-physical}
\end{equation}
For $k=2$, introduce
\begin{equation}
 \mathcal{B}_M(R)
 =M(M-1)-2(M^2-1)R+M(M+1)R^2.
 \label{eq:calBdef}
\end{equation}
Then
\begin{equation}
 \boxed{
 \E\,s_{(1,1)}(\QR)
 =
 \frac{M\tauD^2}{2(M+1)(1-R)^2}
 \left[\mathcal{B}_M(R)-2R^{M+1}\right].
 }
 \label{eq:s11}
\end{equation}

\section{Exact recovery of the Wigner-delay second moment}
\label{sec:wignercheck}

The identity
\begin{equation}
 (\Tr \QR)^2=s_{(2)}(\QR)+s_{(1,1)}(\QR)
 \label{eq:p1-square}
\end{equation}
combines Eqs.~\eqref{eq:s2} and \eqref{eq:s11}.  Because $\tauW=M^{-1}\Tr \QR$, straightforward simplification for $M\ge2$ gives
\begin{equation}
 \boxed{
 \frac{\E\,\tauW^2}{\tauD^2}
 =
 1+\frac{2(1-R^{M+1})}
 {(1-R)^2(M^2-1)}.
 }
 \label{eq:wigner-second}
\end{equation}
This is the known exact finite-$M$ RMT result \cite{fyodorov1996,fyodorov1997,novaes2026}.  Its role here is an independent check of the Schur-level construction.

The comparison with semiclassics is especially informative.  The recent review \cite{novaes2026} notes that the semiclassical calculation yields the same expression without the $R^{M+1}$ contribution.  Equations~\eqref{eq:H-fc} and \eqref{eq:Inpcorr} show that this term is not an isolated accident of the Wigner-delay variance: it is the first member of a complete finite-channel sector that appears already at the level of one-row Schur moments.

\section{Complement symmetry}
\label{sec:complement}

A second exact structure follows directly from the disk automorphism underlying the Poisson kernel.  Such Poisson-kernel transformations and their symmetric-function reproducing properties have a substantial earlier RMT literature; in particular, Baker, Forrester, and Pearce used Jack-polynomial methods and a unimodular mapping to analyze the Poisson-kernel ensemble \cite{baker1998}, while the corresponding scattering-matrix change of variables also appears explicitly in the nonideal-contact formulation \cite{grabsch2018}.  We use this established disk automorphism to derive a consequence for the tunnel-delay barrier factors.  Consider
\begin{equation}
 w=\frac{z+r}{1+rz},
 \qquad |z|=|w|=1.
 \label{eq:mobius}
\end{equation}
On the unit circle its angular Jacobian is $p_R(z)$:
\begin{equation}
 \frac{\dd\arg w}{\dd\arg z}=p_R(z).
 \label{eq:mobius-jac}
\end{equation}
Apply the map to each CUE eigenvalue, writing $\bm z=(z_1,\ldots,z_M)$ and $\bm w=(w_1,\ldots,w_M)$.  Combining the one-body Jacobians with the transformation of the Vandermonde gives
\begin{equation}
 \dd\mu_{\CUE}(\bm w)
 =\left[\prod_{j=1}^{M}p_R(z_j)\right]^M\dd\mu_{\CUE}(\bm z)
 =\bigl[\det\PR(U)\bigr]^M\dd\mu_{\CUE}(\bm z),
 \label{eq:CUE-mobius}
\end{equation}
where $U$ has eigenvalues $z_1,\ldots,z_M$.

For a partition $\lambda\subseteq M^M$, define its complement within the $M\times M$ square by
\begin{equation}
 \lambda^{\square}
 =(M-\lambda_M,M-\lambda_{M-1},\ldots,M-\lambda_1),
 \label{eq:complement-def}
\end{equation}
with trailing zero parts appended to $\lambda$ as necessary.

\begin{theorem}[Square-complement symmetry]
\label{thm:complement}
For every $\lambda\subseteq M^M$,
\begin{equation}
 \boxed{
 I_\lambda(R,M)=I_{\lambda^{\square}}(R,M).
 }
 \label{eq:complement-I}
\end{equation}
Equivalently,
\begin{equation}
 \boxed{
 \frac{[M]_\lambda}{(M\tauD)^{|\lambda|}}
 \E\,s_\lambda(\QR)
 =
 \frac{[M]_{\lambda^{\square}}}
 {(M\tauD)^{M^2-|\lambda|}}
 \E\,s_{\lambda^{\square}}(\QR).
 }
 \label{eq:complement-physical}
\end{equation}
\end{theorem}

\begin{proof}
For this proof, retain the sign of the reflection amplitude and write
\begin{equation}
 P_r(U)=(1-r^2)(\id+rU)^{-1}(\id+rU^\dagger)^{-1},
 \qquad
 J_\lambda(r,M)=\E_{\CUE(M)}s_\lambda(P_r(U)).
 \label{eq:signed-barrier-average}
\end{equation}
Thus $J_\lambda(\sqrt R,M)=I_\lambda(R,M)$, while Haar invariance under $U\mapsto-U$ gives $J_\lambda(-r,M)=J_\lambda(r,M)$.

Let $W$ denote the unitary matrix whose eigenvalues are $w_j=(z_j+r)/(1+rz_j)$.  Equation~\eqref{eq:CUE-mobius} gives
\begin{equation}
 \dd\mu_{\CUE}(\bm z)
 =\bigl[\det P_r(U)\bigr]^{-M}\dd\mu_{\CUE}(\bm w).
 \label{eq:CUE-mobius-inverse-measure}
\end{equation}
The inverse map is obtained by $r\mapsto-r$.  Eigenvalue by eigenvalue,
\begin{equation}
 p_r(z)=p_{-r}(w)^{-1},
 \label{eq:P-inverse-map}
\end{equation}
where $p_r(z)=(1-r^2)/[(1+rz)(1+rz^{-1})]$.  Hence $P_r(U)=P_{-r}(W)^{-1}$.  With $Y=P_{-r}(W)$,
\begin{align}
 J_\lambda(r,M)
 &=\int s_\lambda(Y^{-1})\det(Y)^M\,\dd\mu_{\CUE}(\bm w) \\
 &=\int s_{\lambda^{\square}}(Y)\,\dd\mu_{\CUE}(\bm w),
 \label{eq:complement-proof-int}
\end{align}
where the second equality uses
\begin{equation}
 \det(Y)^M s_\lambda(Y^{-1})=s_{\lambda^{\square}}(Y).
 \label{eq:schur-complement-id}
\end{equation}
The right-hand side is $J_{\lambda^{\square}}(-r,M)=J_{\lambda^{\square}}(r,M)$, proving Eq.~\eqref{eq:complement-I}.
\end{proof}

Equation~\eqref{eq:complement-I} is a consequence of the established Poisson-kernel/Möbius structure \cite{baker1998} applied to the barrier Schur averages.  We have not found this particular square-complement relation stated for the tunnel-delay problem, but we make no broader priority claim for the underlying transformation.

\section{Relation to the semiclassical Schur conjectures}
\label{sec:conjectures}

Novaes and Kuipers \cite{novaeskuipers2023} conjectured three particularly simple structures for tunnel-barrier Schur moments: a transpose relation under $M\mapsto-M$, a reciprocity relation under $R\mapsto R^{-1}$, and barrier independence for self-conjugate partitions.  The 2026 review lists these relations as open and notes that none had then been proved by RMT or semiclassics \cite{novaes2026}.

Our exact finite-$M$ result shows that it is useful to distinguish the perturbative $1/M$ sector from the full RMT answer.  For compactness, define
\begin{equation}
 \mathcal{M}_\lambda(R,M):=\E\,s_\lambda(\QR),
 \label{eq:Mlambda-def}
\end{equation}
and use the same notation for the formal analytic continuation in $R$ that appears in the reciprocity relation.  The proposed $R$-reciprocity then reads
\begin{equation}
 [M]^\lambda\,\mathcal{M}_{\lambda'}(R^{-1},M)
 =
 [M]_\lambda\,\mathcal{M}_\lambda(R,M).
 \label{eq:reciprocity-conj}
\end{equation}
For $\lambda=(2)$, substitution of Eqs.~\eqref{eq:s2} and \eqref{eq:s11} yields
\begin{align}
 &[M]^{(2)}\mathcal{M}_{(1,1)}(R^{-1},M)
 -[M]_{(2)}\mathcal{M}_{(2)}(R,M)
 \nonumber\\
 &\qquad=
 \boxed{
 \frac{M^2\tauD^2}{(1-R)^2}
 \left(R^{M+1}-R^{1-M}\right).
 }
 \label{eq:reciprocity-breaking}
\end{align}
The perturbative pieces cancel; only the finite-$M$ sector beyond all algebraic orders in $1/M$ remains.  Because $R^{-1}>1$ lies outside the physical tunnel-barrier interval, Eq.~\eqref{eq:reciprocity-conj} is best regarded as a formal analytic symmetry.  Equation~\eqref{eq:reciprocity-breaking} shows that it is not an exact identity of the full finite-$M$ RMT answer in this simplest nontrivial case, even though it survives perturbatively.

The self-conjugate conjecture is different.  It predicts
\begin{equation}
 \E\,s_\lambda(\QR)=\E\,s_\lambda(\Qzero),
 \qquad \lambda=\lambda'.
 \label{eq:selfconj}
\end{equation}
Within the present formulation this is equivalent to the pure CUE identity
\begin{equation}
 I_\lambda(R,M)=s_\lambda(1^M),
 \qquad \lambda=\lambda'.
 \label{eq:selfconj-CUE}
\end{equation}
Equation~\eqref{eq:selfconj-CUE} sharply isolates the remaining problem from the inverse-Laguerre ensemble.  The rational-Toeplitz generating function in Sec.~\ref{sec:toeplitz} provides an exact finite-$M$ route to testing and, potentially, proving it for arbitrary $\lambda$.  We leave a proof for all self-conjugate partitions open rather than infer it from low-degree evidence.

\section{Discussion}
\label{sec:discussion}

Several conceptual points follow from the exact reduction.

First, the nonideal-coupling problem separates cleanly into two pieces.  The inverse-Laguerre average carries the ideal delay-time statistics, whereas a CUE Poisson-kernel factor carries the tunnel-barrier dependence.  This separation is much less transparent when one works directly with the nonideal $Q$ distribution.

Second, the $R^{M+1}$ contribution has a transparent mathematical origin.  In the exact CUE ratio representation it is the finite-rank swap term in Eq.~\eqref{eq:row-ratio-specialized}; equivalently, the Toeplitz determinant retains finite-size information that lies beyond all algebraic orders in a formal $1/M$ expansion.  In the particular case of $(\Tr\PR)^2$, the same effect can be seen from the exact CUE trace covariance
\begin{equation}
 \E_{\CUE(M)}\,\Tr U^n\Tr U^{-m}
 =\delta_{nm}\min(n,M),
 \label{eq:CUE-trace-cov}
\end{equation}
whose saturation at $n=M$ produces the geometric tail proportional to $R^{M+1}$.  The Toeplitz calculation extends that finite-size mechanism to Schur observables at all orders.

Third, the generating-function formulation organizes the auxiliary complexity by partition width.  The original Toeplitz determinant is $M\times M$, but the number of zeros and poles of its rational symbol is set by $L$.  Classical exact ratio/Toeplitz formulas \cite{day1975,conrey2005,bumpgamburd2006} can therefore reorganize the calculation in terms of these auxiliary factors.  For a fixed target partition one may choose $L=\lambda_1$ while $M$ enters parametrically.  This retains the inverse-Laguerre/symmetric-function and determinant-generating-function perspective explored previously at ideal coupling \cite{araujo2017thesis}; the ideal input is explicit in Eq.~\eqref{eq:ideal-schur}, and the nonideal construction is derived in Secs.~\ref{sec:factorization}--\ref{sec:rows}.

Finally, the present analysis clarifies how semiclassical conjectures should be interpreted.  Their perturbative content can be correct to every algebraic order in $1/M$ while the exact finite-channel answer contains exponentially small additions.  In particular, the missing $R^{M+1}$ term in the Wigner-delay second moment \cite{novaes2026} is the lowest-order manifestation of the exact sector in Eq.~\eqref{eq:Inpcorr}.

\section{Conclusion}
\label{sec:conclusion}

We derived an exact finite-channel framework for Schur moments of the Wigner--Smith time-delay matrix in a chaotic cavity with broken time-reversal symmetry and a uniform tunnel barrier.  The exact nonideal-to-ideal mapping reduces every Schur moment to an ideal inverse-Laguerre factor multiplied by a CUE barrier average.  A dual-Cauchy generating function then turns the barrier problem into a rational Toeplitz determinant, bringing established exact CUE ratio formulas directly into the time-delay problem.

Two partition families can be evaluated explicitly at arbitrary allowed order.  For one-row partitions, the closed generating function \eqref{eq:H-exact} yields all finite delay-time moments with $n\le M$ and exposes a finite-channel sector proportional to $R^{M+1}$.  One-column moments are given by the finite sum \eqref{eq:column-barrier}.  Together they reproduce the exact Wigner-delay second moment including the exponentially small correction missed by perturbative semiclassics.  We also obtained an exact complement symmetry, Eq.~\eqref{eq:complement-I}, and showed explicitly that the proposed $R\leftrightarrow R^{-1}$ reciprocity for $\lambda=(2)$ is a perturbative-sector symmetry rather than an identity of the complete finite-$M$ answer.

The characteristic-polynomial ratio identities used in the derivation are classical \cite{conrey2005,bumpgamburd2006}.  Our novelty claims concern the reduction of the tunnel-barrier Wigner--Smith problem to those identities, the resulting finite-channel Schur-moment formulas, and the interpretation of the sector missing from semiclassics---not the underlying CUE ratio theorem itself.

The most immediate open problem is the self-conjugate identity \eqref{eq:selfconj-CUE}.  The present CUE--Toeplitz representation removes the inverse-Laguerre part of that conjecture and may make a general proof accessible.  A second direction is a general-$\beta$ tunnel-barrier extension.  Since the ideal factor is already available from Jack-polynomial methods \cite{novaes2022}, this would test whether the broader 2017 inverse-Laguerre generating-function framework \cite{araujo2017thesis} can handle the genuinely nonideal barrier factor beyond $\beta=2$.  Both directions can be pursued without altering the exact finite-$M$ results established here.

\appendix

\section{Derivation of the one-row generating function}
\label{app:rowproof}

Starting from Eq.~\eqref{eq:H-def} and the Heine--Szeg\H{o} identity,
\begin{equation}
 H_M(t)=D_M[f_t],
 \qquad
 f_t(z)=\frac{1}{1-tp_R(z)}.
 \label{eq:Ht-toeplitz}
\end{equation}
Equation~\eqref{eq:x-def} is equivalent to
\begin{equation}
 1+R-t(1-R)=x^{-1}+Rx.
 \label{eq:x-identity}
\end{equation}
Hence
\begin{equation}
 f_t(z)
 =
 x\,
 \frac{(1+rz)(1+rz^{-1})}
 {(1+rxz)(1+rxz^{-1})}.
 \label{eq:f-factor}
\end{equation}
Set
\begin{equation}
 a=rx,
 \qquad
 \eta=\frac{(x-1)(1-Rx)}{1-Rx^2}.
 \label{eq:a-eta}
\end{equation}
Using
\begin{equation}
 \frac{1}{(1+az)(1+az^{-1})}
 =\frac{1}{1-a^2}
 \sum_{m\in\mathbb Z}(-a)^{|m|}z^m,
\end{equation}
one finds that the Toeplitz matrix of $f_t$ is simply
\begin{equation}
 T_M(f_t)=\id+\eta K_M(a),
 \qquad
 K_M(a)_{ij}=(-a)^{|i-j|}.
 \label{eq:Tf-K}
\end{equation}
The Kac--Murdock--Szeg\H{o} matrix obeys
\begin{equation}
 \det K_M(a)=(1-a^2)^{M-1}
 \label{eq:detK}
\end{equation}
and
\begin{equation}
 K_M(a)^{-1}
 =\frac{1}{1-a^2}
 \begin{pmatrix}
 1&a&&&\\
 a&1+a^2&a&&\\
 &\ddots&\ddots&\ddots&\\
 &&a&1+a^2&a\\
 &&&a&1
 \end{pmatrix}.
 \label{eq:Kinv}
\end{equation}
Therefore
\begin{equation}
 H_M(t)
 =\frac{\det C_M}{1-a^2},
 \label{eq:HM-C}
\end{equation}
where $C_M$ is tridiagonal with off-diagonal entries $a$, endpoint diagonal entries
\begin{equation}
 b_{\mathrm{end}}=1+\eta(1-a^2)=x(1+R-Rx),
\end{equation}
and interior diagonal entries
\begin{equation}
 b_{\mathrm{int}}=1+a^2+\eta(1-a^2)=x(1+R).
\end{equation}
Let $F_n$ be the determinant of the first $n$ rows and columns before the final endpoint replacement.  Then
\begin{equation}
 F_n=b_{\mathrm{int}}F_{n-1}-a^2F_{n-2},
 \qquad
 F_0=1,
 \quad F_1=b_{\mathrm{end}}.
\end{equation}
The characteristic roots are $x$ and $Rx$, so
\begin{equation}
 F_n=
 \frac{Rx-1}{R-1}x^n
 +\frac{R(1-x)}{R-1}(Rx)^n.
 \label{eq:Fn}
\end{equation}
The last endpoint gives
\begin{equation}
 \det C_M=b_{\mathrm{end}}F_{M-1}-Rx^2F_{M-2}
 =
 \frac{x^M}{1-R}
 \left[(1-Rx)^2-R^{M+1}(x-1)^2\right].
\end{equation}
Using $1-a^2=1-Rx^2$ proves Eq.~\eqref{eq:H-exact}.

\subsection{Coefficient extraction}

Set $y=x-1$.  Equation~\eqref{eq:x-def} becomes
\begin{equation}
 y=t\,\frac{1+y}{1-cy},
 \qquad c=\frac{R}{1-R}.
 \label{eq:y-Lagrange}
\end{equation}
The Lagrange--B\"urmann formula applied separately to the two terms in Eq.~\eqref{eq:H-split} gives Eqs.~\eqref{eq:Ipert} and \eqref{eq:Inpcorr}.  In particular, the correction begins only at order $t^2$, which explains the barrier independence of the mean delay.

\section{Summation of the one-column principal minors}
\label{app:columnproof}

For a principal submatrix indexed by $i_1<\cdots<i_k$, the Markov structure of the Kac--Murdock--Szeg\H{o} matrix gives Eq.~\eqref{eq:KMSminor}.  Writing the gaps as
\begin{equation}
 d_a=i_{a+1}-i_a=1+g_a,
 \qquad g_a\ge0,
\end{equation}
we obtain
\begin{align}
 I_{(1^k)}
 &=
 (1-R)^{k-1}
 \sum_{g_1,\ldots,g_{k-1}\ge0}
 \left(M-k+1-\sum_a g_a\right)_+
 \prod_{a=1}^{k-1}
 \left(1+R+\cdots+R^{g_a}\right),
 \label{eq:gap-sum}
\end{align}
where $(q)_+=\max(q,0)$.  Introduce a bookkeeping variable $z$.  Since
\begin{equation}
 \sum_{g\ge0}z^g(1+R+\cdots+R^g)
 =\frac{1}{(1-z)(1-Rz)},
\end{equation}
the gap sum can be performed by coefficient extraction.  Expanding the two geometric factors and summing the allowed starting positions gives
\begin{equation}
 I_{(1^k)}
 =
 (1-R)^{k-1}
 \sum_{j=0}^{M-k}
 \binom{M-j}{k}
 \binom{k+j-2}{j}R^j,
\end{equation}
which is Eq.~\eqref{eq:column-barrier}.

\section{A useful \texorpdfstring{finite-$M$}{finite-M} trace check}
\label{app:tracecheck}

From Eq.~\eqref{eq:poisson-fourier},
\begin{equation}
 \Tr\PR
 =M+
 \sum_{n\ge1}(-\sqrt R)^n
 \left(\Tr U^n+\Tr U^{-n}\right).
\end{equation}
The exact CUE identity
\begin{equation}
 \E_{\CUE(M)}\,\Tr U^n\Tr U^{-m}
 =\delta_{nm}\min(n,M)
\end{equation}
therefore gives
\begin{align}
 \E_{\CUE(M)}\,(\Tr\PR)^2
 &=M^2+2\sum_{n\ge1}\min(n,M)R^n
 \\
 &=M^2+\frac{2R(1-R^M)}{(1-R)^2}.
 \label{eq:traceP2}
\end{align}
Likewise,
\begin{equation}
 \E_{\CUE(M)}\,\Tr\PR^2=M\frac{1+R}{1-R}.
\end{equation}
Using $s_{(2)}=(p_1^2+p_2)/2$ and $s_{(1,1)}=(p_1^2-p_2)/2$ reproduces Eqs.~\eqref{eq:s2} and \eqref{eq:s11} independently.  This makes the origin of $R^{M+1}$ especially transparent: it is generated by the finite-$M$ saturation $\min(n,M)$ of CUE trace correlations.

\ack{The author used OpenAI ChatGPT (GPT-5.6 Sol) to assist with manuscript structuring, language editing, literature-search support, and preparation of submission materials. All AI-assisted material was critically reviewed by the author, who takes full responsibility for the scientific content, derivations, references, and final text.}

\data{No new data were created or analysed in this study.}

\bibliographystyle{unsrtnat}
\bibliography{references}

\end{document}